\documentclass[10pt, conference, compsocconf]{IEEEtran}
\usepackage{times,epsfig}
\usepackage{amsmath}
\usepackage{graphicx}
\usepackage{amsfonts,epsfig,subfigure}
\usepackage{amssymb,graphicx}
\usepackage{booktabs,ctable}
\usepackage[inoutnumbered,linesnumbered,algoruled,slide,vlined]{algorithm2e}
\usepackage{algorithmic}

\newtheorem{theorem}{Theorem}
\newenvironment{proof}[1][Proof]{\begin{sloppypar}\noindent\textit{#1.} }
{\vspace*{3mm}\hfill$\Box$\end{sloppypar}}
\newtheorem{lemma}{Lemma}
\newtheorem{definition}{Definition}

\begin{document}

\title{Energy-Optimal Scheduling in Low Duty Cycle Sensor Networks }

\author{\IEEEauthorblockN{Nursen Aydin, Mehmet Karaca and Ozgur Ercetin}
\IEEEauthorblockA{Faculty of Engineering and Natural Sciences,
Sabanci University, Istanbul, Turkey\\
\{nursenaydin,mehmetkrc,oercetin\}@sabanciuniv.edu} }

\maketitle
\begin{abstract}
Energy consumption of a wireless sensor node mainly depends on the
amount of time the node spends in each of the high power active (e.g.,
transmit, receive) and low power sleep modes.  It has been well
established that in order to prolong node's lifetime the duty-cycle
of the node should be low.  However, low power sleep modes usually have
low current draw but high energy cost while switching to the active mode with
a higher current draw.  In this work, we investigate a MaxWeight-like
opportunistic sleep-active scheduling algorithm that takes into account time-
varying channel and traffic conditions. We show that our algorithm is energy
optimal in the sense that the proposed ESS algorithm can achieve
an energy consumption which is arbitrarily close to the global
minimum solution. Simulation studies are provided to confirm the theoretical results.
\end{abstract}


\section{Introduction}
Wireless sensor network consists of many battery operated nodes with
limited processing and wireless communicating abilities. They are
used in many different areas such as military, scientific research
or medical diagnostics. Since sensor networks are uniquely
identified by their requirement of operation for a long period
without outside intervention, energy consumption is a primary
concern. Unnecessary energy consumption can be caused from the
implementation at PHY layer. For instance, keeping the sensor nodes
in active all the time is one of the main  inefficiency.
Furthermore, energy unaware scheduling algorithms employed at MAC or
above layers is the another reason for the unnecessary energy
consumption.

A particularly important strategy is to minimize wasted energy due
to idle listening, in other words, operate the nodes with low duty
cycles~\cite{survey}. Duty cycle is defined as the proportion of
time the node stays active in its lifetime. Thus, it is necessary to
reduce the duty cycle in order to resolve this conflict. A sensor
node consumes energy while transmitting, receiving and sleeping
Also, switching between these three modes is the another issue and
the sensors require extra time and spend extra energy for switching.
In most prior works, it was assumed that the switching energy is
negligible compared to the energy spent in other modes. It is
recently observed that, this assumption does not hold in low duty
cycle sensor networks~\cite{ruzzelli1}. For example,
in~\cite{ruzzelli1} for \texttt{CC1010} type sensor nodes, while the
switching energy from sleep mode to transmit mode is given as 47.75
$\mu J$ for the switching time is 0.7 ms, the energy consumed during
the transmission taken 0.7 ms is equal to 124.8$\mu J$. This
measurement shows that the switching energy is crucial and should be
taken into account for any energy efficient scheduling algorithm.

In the literature, there are many scheduling algorithms that have
been developed for network stability problem with strong theoretical
results. Tassiulas et al. in~\cite{Tass} first introduced the design
of backpressure algorithms based on Lyapunov drift techniques to
achieve network stability and the authors in~\cite{neely} developed
algorithm which deals with performance optimization and queue
stability problems simultaneously in a unified framework. There
exists also a rich literature on using \emph{Lyapunov drift
optimization} for solving problems of energy optimization in
wireless networks. In this paper, one of our objectives is to show
the significance of switching energy  in the design of scheduling
policy.

\section{Related Works}
There are many works that offer low energy consumption in the literature~%
\cite{related2},~\cite{related3},~\cite{related4},~\cite{related5}.
In all these works, the proposed algorithms and experiments have
been designed by only taking into account the energy consumption in
the active and sleep modes but they have not considered the consumed
energy during switching from one state to another. In addition, none
of these works deals with the
network stability problem with the objective of minimum energy. In~\cite%
{related1}, the authors have proposed wireless sensor network
protocols that take into account the source of energy consumptions
in a network simulation model. They have also mentioned about
switching energy but leaved it as an open research area. The study
in~\cite{related6} have considered switching energy cost since it
has been observed that significant energy consumption occurs when
switching from sleep mode to the active mode. On the other hand, in
some studies like~\cite{related7} it has been stated that the energy
cost of switching is small. Generally, it is common to ignore
switching energy in the literature. There are also various
approaches proposed for designing duty cycling
protocols~\cite{polastre},~\cite{ruzzelli1},~\cite{gu}, for
sleep mode protocols~\cite{jurdak} and for routing protocols~\cite{ruzzelli}%
,~\cite{chang}. All of these works try to maximize the lifespan or
the utility of the network without considering the network stability
problem.

The works in~\cite{sridharan},~\cite{recharge} and~\cite{song} are
the most important ones considering network stability issue for
wireless sensor network. In~\cite{sridharan}, the authors proposed a
back pressure algorithm designed using a Lyapunov drift based
optimization framework using the receiver capacity model. However,
they did not aim to minimize the average energy consumption. An
optimal control algorithm for rechargeable wireless sensor nodes is
proposed in~\cite{recharge} where the objective function is maximize
the network utility  subject to stability. In~\cite{song}, the
authors proposed a scheduling algorithm based on Lyapunov
optimization theory in order to minimized the consumed energy while
maintaining the network stability. However, they do not consider low
duty cycle sensor networks.

In this paper, we develop a throughput optimal scheduling algorithm
based on the Lyapunov drift framework which considers not only
minimizing energy but also sustaining network stability for low duty
cycle sensor networks. Our energy switching and scheduling (ESS)
algorithm is energy optimal in the sense that it minimizes the
overall expected energy consumed in the network over all scheduling
algorithms and remains throughput optimal. The proposed ESS
algorithm prolongs the lifetime compared to the benchmark algorithm
which does not consider switching energy. Therefore, our proposed
ESS algorithm is more favorable for scheduling in low duty cycle
wireless sensor networks.

The rest of the paper is organized as follows. Section III describes
the system model and problem formulation. The proposed ESS algorithm
is introduced in Section IV. Lyapunov optimization technique which
is used to show the performance and the optimality of ESS algorithm
and a distributed algorithm are given in Section V. We present our
simulation result in Section VI and Section VII concludes this
paper.

\section{System Model and Problem Statement}
\subsection{System Model}
We consider an uplink scenario where $N$ sensor nodes sense the
environment and transmit to a base station. The network is operated
in slotted time $t = 0, 1, 2,...$ where slot $t$ corresponds to the
time interval $[t,t+1)$. The node channels are time-varying and the
instantaneous channel state is represented by $\mathbf{M}(t)$ and
assume it is i.i.d distributed over a finite set $\mathcal{M}$,
$\mathbf{M}(t) \in \mathcal{M}$. We also define $\pi_m \triangleq
Pr(\mathbf{M}(t)=m)$ as steady state probability of being in channel
state $m$ at time slot $t$. Channels are assumed to hold their state
within one time slot. The data rate of node $n$ during slot $t$ is
represented by $\mu_n(I_n(t),M_n(t))$ (in units of packets/slot)
where $M_n(t)$ is the channel state of node $n$ and $I_n(t)=1$,
whenever node $n$ is selected as a transmitting node at time $t$,
and $I_n(t)=0$ otherwise. We assume that $\mu_n(I_n(t),M_n(t))$ is
upper bounded with $\mu_n(I_n(t),M_n(t)) \leq \mu_{max}$. For ease
of notation, we use $\mu_n(t)$ instead of $\mu_n(I_n(t),M_n(t))$ in
the rest of the paper.

At each time slot $t$, the number of new packets generated by node $n$
is denoted by $R_n(t)$. We assume that $R_n(t)$ is i.i.d for each
time slot with an average rate of $\lambda_n$, and it is upper
bounded with $R_n(t) \leq R_{max}$. We define the network admission
rate vector as
$\boldsymbol{\lambda}=\{\lambda_1,\lambda_2,...,\lambda_N\}$ where
$\lambda_n$ is the average exogenous arrival rate of node $n$. Let
$\Lambda$ be the network capacity region, i.e., the set of all
feasible admission rate vectors that the network can support. 
\subsection{Notations}
In our network model, there are $K$ the sensor nodes which can be in
active or sleep modes at a particular time slot. Therefore, the node
$n \in N$ where $N$ represents the set of all nodes has an action
space denoted as $A=\{0,1\}$. We denote $s$ as the action taken by
any node and $s=0$ and $s=1$ represent the sleep and active modes
respectively. The sensor can also switch from active to sleep or
vice versa and we denote $sw \in \{01, 10\}$ as the switching action
where $sw=01$ and $sw=10$ represent the switching from sleep to
active and active to sleep modes respectively. In addition, we
denote the set of all sleep and active nodes as $Sl$ and $Ac$
respectively at a given time slot.
\subsection{Problem Formulation}
We begin with a definition of network stability. Let $Q_n(t)$ be the
queue backlog of node $n$ at time slot $t$.
\begin{definition}
A queue is strongly stable if
\begin{equation}
\limsup_{t\rightarrow \infty}\frac{1}{t}
\sum_{\tau=0}^{t-1}\mathbb{E}(Q_n(t)) < \infty \label{eq:defination}
\end{equation}
\end{definition}
Moreover, if every queue in the network is stable then the network
is called stable. The system stability region is the the closure of
the convex hull of all arrival rate points for which there exists a
feasible scheduling policy that achieves system
stability~\cite{neely}.

We consider an energy consumption model defined as follows. At the
beginning of transmission, each sensor has a full battery of $E$. In
each time slot, a particular node can be either in active or sleep
modes according to its remaining energy, queue backlog and channel
condition. During the sleep mode, the node is unable to transmit
packets. However, packets continue to be generated by the node, and
they are queued until the node has the opportunity to transmit them.
Let $e_s^n(t)$ be the energy consumed by node $n$ when it is in mode
$s \in \{0,1\}$ at time $t$. When the node is active, it consumes an
energy $c_n(t)$ that depends on the current draw of active mode by
the circuit, and how much time the node spends in active mode. Also,
the energy needed for packets transmission should be included to the
total energy expenditure when the node is selected to transmit. Then
at time $t$, the total energy consumption by node $n$ in active mode
becomes $e_1^n(t)=c_n(t)+\mu_n(t) \alpha_n$ where $\alpha_n$ denotes
the energy spent for a packet transmission. Furthermore, if a node
switches from sleep-to-active or active-to-sleep modes, then the
switching energy denoted as $e_{sw}^n(t)$ is consumed. The values of
energy costs are given in Section VI. Hence, the overall energy
spent by node $n$ during time slot $t$ is given by:
\begin{equation}
H^n(t)=\sum_{s\in A} \textrm{e}_s^n (t) + e_{sw}^n (t) \qquad s\in
\{0,1\}. \label{eq:energy-sleep}
\end{equation}
\noindent and the total energy consumption in the network at time
slot $t$ is obtained by summing over all nodes in the system,
\begin{equation}
H(t)=\sum_n \sum_{s\in A}  \textrm{e}_s^n (t) + \sum_n e_{sw}^n (t)
\qquad s\in \{0,1\} \label{eq:energy-total}
\end{equation}
Define  the time average expected total energy consumption as,
\begin{align}
\overline{H}=\limsup_{t\rightarrow\infty}
\frac{1}{t}\sum_{\tau=0}^{t-1}\mathbb{E}\left[\sum_n \sum_s
\textrm{e}_s^n (\tau) + \sum_n e_{sw}^n (\tau) \right]
\end{align}
The expectation is with respect to the randomness that arises from
channel variations and arrival process and possibly from random
stationary switching policy.  The overall energy consumption during
time $t$ is upper bounded by a finite value since all sensors have a
limited battery power and thus,  without lost of generality, we
assume that the following is satisfied at every time slot,
\begin{align}
\sum_n H^n(t)\leq H_{max}, \; \forall t
\end{align}
We are interested in minimizing the total average energy consumption
in the network while keeping the queue sizes of the sensor nodes
bounded. Then, our optimization problem can be formalized as
follows,
\begin{align}
\min \quad & \overline{H} \label{eq:basic-optim}\\
s.t. \quad &\text{network stability.}
\end{align}

\section{Energy-Aware Switching and Scheduling (ESS) Algorithm }
In this section, we introduce our energy aware sleep-active
scheduling algorithm which asymptotically minimizes the average
network energy consumption subject to the network stability.
Furthermore, the proposed algorithm is shown to be throughput
optimal meaning that the algorithm can guarantee the network
stability for all feasible network admission rates.

At each time slot, the proposed Energy-aware Switching and
Scheduling (ESS) algorithm determines the energy modes of operation
for each sensor nodes and selects the transmitting node which maximizes the following:\\
 \noindent \textbf{ESS}\\
\begin{equation}
 \max_{n\in Ac,s} \quad Q_n(t)\mu_n(t)-V \textrm{e}_{sw}^n (t)- V
 \sum_{s\in A}
\textrm{e}_{s}^n (t).
 \label{eq:obj}
 \end{equation}
In \eqref{eq:obj}, $V>0$ is a system parameter and shows the
well-known delay-energy tradeoff~\cite{neely}. For very large value
of $V$ we can push the average energy consumed by ESS algorithm
arbitrarily close to the global minimum energy. However, in that
case in order to maximize~\eqref{eq:obj} nodes stay in the sleep
mode most of the time, and consequently, the queue sizes increase.

\textit{Remark 1}: As it is seen in~(\ref{eq:obj}), ESS algorithm is
different from the Max Weight algorithm, and the algorithm proposed
in~\cite{song} since ESS not only considers the energy consumption
in sleep and active modes but also takes into account the switching
energy cost. It is shown in simulation studies that this cost cannot
be ignored and has significant in the life time of the network.

\textit{Remark 2}: The basic properties of ESS algorithm are as
follows. In each time slot $t$ at most one node is active, and this
node transmits. The active node continues to transmit in consecutive
slots until some other node is the maximizer of~(\ref{eq:obj}). On
the other hand, the rest of the nodes stay in the sleep mode and
their queue backlogs increase. Once the queue backlog of a sleeping
node is sufficiently high in order to maximizes~(\ref{eq:obj}) then
this node is selected to transmit. Furthermore, if all of the queue
backlogs are low, the nodes choose to stay in the sleep mode. Thus,
the system can be in the idle state where no nodes choose to be
active during some of the time slots.
\begin{lemma} ESS satisfies the following properties.
\begin{enumerate}
\item Sensor nodes transmit in a bursty fashion, i.e., transmit
multiple packets once they capture the channel.

\item The system operation is non-work conserving, i.e., there are
idle slots when no sensor node transmits even if their backlog is
non-empty.

\end{enumerate}
\end{lemma}
\begin{proof}
For notational brevity, let $n$ and $m$ represent the active and
sleeping nodes respectively and their corresponding queue sizes are
denoted as $Q_n(t)$ and $Q_{m}(t)$ at time slot $t$. If the active
node $n$ prefers to be again active during next time slot $t+1$, it
does not need to change its energy mode and there will be no
switching energy cost ($e^n_{sw}(t)=0$). Thus, its weight is equal
to $Q_n(t)\mu_n(t)-Ve_1^n(t)$ according to \eqref{eq:obj}. On the
other hand, if it switches its mode from active to the sleep, then
it obtains the weight which is equal to $-Ve_0^n(t)-Ve_{01}^n(t)$
since it cannot transmit in sleeping mode and $\mu_n(t)=0$ and also
it pays a switching energy cost which is equal to $e^n_{10}(t) > 0$.
Therefore, the node $n$ prefers to stay in active mode if and only
if the weight obtained in active mode exceeds the weight resulting
from switching to sleep mode. In other words, the following
inequality should be hold.
\begin{equation*}
-Ve_0^n(t)-Ve_{10}^n(t) < Q_n(t)\mu_n(t)-Ve_1^n(t)
\end{equation*}
Thus
\begin{equation}
\frac{V(e_{1}^n(t)-e_0^n(t)-e_{10}^n(t))}{\mu_n(t)} < Q_n(t)
 \label{eq:queueac}
\end{equation}
Furthermore, since ESS algorithm aims to maximize \eqref{eq:obj},
the node $n$ continues to be active and transmits if and only if it
is the maximizer of \eqref{eq:obj}.
\begin{align}
\max_{b \in Sl} \left[Q_b(t) \mu_b (t)-Ve_{01}^b(t)-Ve_1^b(t)\right]& \nonumber \\
 < Q_n(t) \mu_n(t)-Ve_1^n(t) &
 \label{eq:queueac1}
\end{align}
Therefore, as long as inequalities \eqref{eq:queueac} and
\eqref{eq:queueac1} hold, active sensor node transmits in a bursty
fashion and the other nodes do not change their energy modes. This
completes the first part of Lemma 1.

On the other hand, if the sleeping node $m$ prefers to stay in sleep
mode during next slot $t+1$, according to equation~(\ref{eq:obj})
its weight is $-V e_0^{m}(t)$ since it cannot transmit then
$\mu_{m}(t)=0$ and there will no switching cost $e^{m}_{sw}(t)=0$.
If it switches to active mode and transmit, then it obtains the
weight which is equal to
$Q_{m}(t)\mu_{m}(t)-Ve_{01}^{m}(t)-Ve_1^{m}(t)$. Therefore, the node
$m$  node continues to stay in sleep mode if and only if its queue
sizes is not large enough to cover the reward being in sleep mode.
In other words, it stays in sleep mode as long as the following
inequality holds,
\begin{equation*}
-Ve_0^{m}(t) > Q_{m}(t)\mu_{m}(t)-Ve_{01}^{m}(t)-Ve_1^{m}(t)
\end{equation*}
Thus
\begin{equation}
\frac{V(e_{1}^{m}(t)+e_{01}^{m}(t)-e_0^{m}(t))}{\mu_{m}(t)} > Q_{m}(t)
 \label{eq:queue}
\end{equation}
Similarly, the active node prefers to switch to sleep mode, if its
backlog is not large enough to maximize \eqref{eq:obj}. Therefore,
the system is idle if the following inequalities are satisfied
\begin{align}
Q_n(t) < & \frac{V(e_{1}^n(t)-e_0^n(t)-e_{10}^n(t))}{\mu_n(t)} & \forall n \in Ac \\
Q_{m}(t) < & \frac{V(e_{1}^{m}(t)+e_{01}^{m}(t)-e_0^{m}(t))}{\mu_{m}(t)}
& \forall {m} \in Sl
\end{align}
\end{proof}
\section{Throughput optimality of ESS}
Despite the fact that \eqref{eq:basic-optim}-(7) is a convex
optimization problem, a direct solution is generally overly
complicated since the arrival rates resulting network stability does
not admit in general a simple characterization. Fortunately, we can
use the framework of \cite{neely}, and obtain a dynamic scheduling
policy that operates arbitrarily closely to the optimal point. This
is obtained in two steps: first, a dynamic scheduling policy that
achieves the stability of transmission queues whenever the arrival
rates are inside the stable network operating region is obtained.
Second, we define a stationary randomized algorithm and compare the
randomized policy with our ESS algorithm.
\subsection{Optimal Stationary Randomized Algorithm}
Now, we present a stationary randomized algorithm which is used to
prove the stability property and throughput optimality of ESS
algorithm. The randomized algorithm schedules the nodes according to
a stationary, and possibly randomized function of only the data
rates and it is independent of the queue backlog. Let us denote the
stationary randomized algorithm as $RND$. With $RND$ algorithm, the
switching decision is determined by a Markovian process with two
states; sleep and active. We denote the transition probabilities
from active to sleep states and from sleep to active states as
$p_{10}$ and $p_{01}$ respectively. The steady-state probabilities
of the active and sleep modes are denoted as $\pi_{a}$ and $\pi_{s}$
respectively. We assume that channel processes are ergodic with
steady-state probabilities $\pi_{k}$ where $k$ is describing the
channel state in a finite set, ${k}\in \{1,2,...,M\}$. The following
theorem specifies the minimum energy required for stability, among
the class of stationary policies that randomly decides on
sleep-active modes and transmit scheduling.
\begin{theorem}
If arrival rate vector is strictly in the capacity region
$\boldsymbol{\lambda} \in \Lambda$, then the minimum energy that can
stabilize the network is denoted by  $h^*$ and it is equal to the
solution of the following optimization problem that is defined in
terms of transition probabilities $(p^n_{ij})_k$ for $(i,j) \in
\{0,1\}$, steady-state probabilities $(\pi_l^n)_k$ for $l\in
\{a,s\}$ and transmission probabilities $(\pi_{tr}^n)_k$ at channel
state $k \in \mathbf{M}$.
 \begin{align*}
\min \quad & \sum_{k \in \mathbf{M}} \pi_k \sum_{n}(p_{10}^n)_k
(\pi_s^n)_k + (p_{01}^n)_k (\pi_a^n)_k e_{sw}^n + (\pi_s^n)_k e_0^n\\
& \qquad \quad \quad \; +(\pi_a^n)_{k} (c^n(t)+(\pi_{tr}^n)_k(\mu_n)_k\alpha_n) \\
s.t. \quad & \sum_{k \in \mathbf{M}} \pi_k (\pi_a^n)_{k} (\pi_{tr}^n)_k(\mu_n)_{k} \geq \lambda_n \quad \forall n\\
    \quad & \sum_k (\pi_p^n)_k \leq 1  \qquad \; \; \; \forall p
    \in \{a,s\},\quad \forall n\\
    \quad &\sum_{n \in Ac} (\pi_{tr}^{n})_k  \leq 1 \quad \qquad \quad \forall k
\end{align*}
where $(\mu_n)_k$ is the data rate of node $n$ at channel state
$k \in \mathbf{M}$.
\end{theorem}
\begin{proof} The proof follows the same line of logic as the proof
given~\cite{neely2}. The optimization problem above can be solved
since the set of channel state is finite and the energy set is
compact. Thus, we can achieve the minimum energy by a randomized
algorithm which decides the switching decision with probabilities
$(p^n_{ij})_k$ and selects the transmission node with probability $(\pi_{tr}^n)_k$.

Since $\boldsymbol{\lambda}$ lies in the interior of
the network capacity region $\Lambda$, it follows that there exists
a small positive $\boldsymbol{\epsilon} > 0$ such that
$\boldsymbol{\lambda}+ \boldsymbol{\epsilon} \in \Lambda$. Then,
there exist a stationary policy which stabilizes the network while
providing a data rate of
\begin{align}
\mathbb{E}(\boldsymbol{\hat{\mu}}(t)) \geq \boldsymbol{\lambda} +
\boldsymbol{\epsilon}, \label{eq:rnd_mu}
\end{align}
\noindent where $\boldsymbol{\hat{\mu}}$ is the transmission rate
induced by the randomized policy. Let $h^*(\epsilon)$ be the minimum
energy consumed by such policy. Since expected transmission rate is
greater than the arrival rate, the network is stable. In addition,
since $h^*(\epsilon)$ is the minimum energy required to stabilize
the network with arrival rate $\boldsymbol{\lambda} +
\boldsymbol{\epsilon}$, any mixed strategy results in higher energy
consumption. Therefore, it is obvious that
\begin{align}
h^* \leq  h^* (\epsilon) \leq \left( 1-
\frac{\epsilon}{\epsilon_{max}} \right) h^* +
\frac{\epsilon}{\epsilon_{max}}  h^* (\epsilon_{max}) \label{eq:gav}
\end{align}

\noindent where $\epsilon_{max}$ is the largest scalar where
$\lambda_n + \epsilon_{max} \in \Lambda$. When
$\epsilon=\epsilon_{max}$, nodes consume much more energy to
stabilize the queues. Therefore, $h^* (\epsilon_{max}) \geq h^*
(\epsilon) \; \forall \epsilon$. It follows that $h^* (\epsilon)
 \rightarrow h^*$ as $\epsilon \rightarrow 0$.\\
\end{proof}
 Theorem 1 indicates that the minimum energy consumption for stability is
achieved among the class of stationary policies that measure the
current channel state $\mathbf{M}(t)$, and then randomly decides
switching and selects the user which transmits.
\subsection{Analysis of ESS}
In our work, we use Lyapunov drift and optimization tools
~\cite{neely} to show the optimality of ESS algorithm. The advantage
of this tool is the ability to deal with performance optimization,
and queue stability problems simultaneously in a unified framework.
For node $n$ the queue dynamics is given by
\begin{equation}
Q_n(t+1)= \max(Q_n(t)-\mu_n(t),0)+R_n(t) \label{eq:queue_dyn}
\end{equation}
We can write the following inequality by using the fact $([a]^+)^2
\leq (a)^2, \quad \forall a$:
\begin{align}
Q^2_n(t+1) \leq &Q^2_n(t)\\&+(\mu_{max})^2+(R_{max})^2 -\nonumber
2Q_n(t)[\mu_n(t)-R_n(t)] \label{eq:queuek}
\end{align}
Let $\textbf{Q}(t)=(Q_1(t),Q_2(t),...,Q_N(t))$  be a queue vector
and define the following quadratic Lyapunov function
\begin{equation}
L(\mathbf{Q}(t))\triangleq \frac{1}{2} \sum_{n=1}^N(Q_n^2(t))
\end{equation}
and the conditional Lyapunov drift is
\begin{equation}
\Delta(t)\triangleq \mathbb{E}
\left\{L(\mathbf{Q}(t+1))-L(\mathbf{Q}(t))|\mathbf{Q}(t) \right\}
\end{equation}
Then, by using (17) Lyapunov drift of the system
satisfies the following inequality at every time slot,
\begin{equation}
\Delta(t) \leq B -
\sum_n\mathbb{E}\left\{Q_n(t)(\mu_n(t)-R_n(t))|\mathbf{Q}(t)\right\},\label{eq:Lyp_drift}
\end{equation}
where
\begin{equation}
B=\frac{K}{2} ((\mu_{max})^2+(R_{max})^2).
\end{equation}
In addition, we define a cost function
$\mathbb{E}(H(t)|\mathbf{Q}(t))$ as the expected total energy
consumption during time slot $t$. After adding the cost function
multiplied by $V$ to both sides of (\ref{eq:Lyp_drift}), we have the
following.
\begin{align}
&\Delta (t) +  V \mathbb{E}(H(t)|\mathbf{Q}(t)) \nonumber\\
& \leq B -\sum_n\mathbb{E}\left\{Q_n(t)(\mu_n(t)-R_n(t))|\mathbf{Q}(t)\right\}+V
\mathbb{E}(H(t)|\mathbf{Q}(t))\nonumber\\
&\leq B - \sum_n\mathbb{E}\left\{Q_n(t)\mu_n(t)|\mathbf{Q}(t) \right\} + V \sum_n
\sum_s \mathbb{E}\left\{\textrm{e}_s^n (t) |\mathbf{Q}(t) \right\}\nonumber\\
&+V \sum_n
\mathbb{E}\left\{\textrm{e}_{sw}^n (t) |\mathbf{Q}(t) \right\}\nonumber\\
& \leq B-\underbrace{\sum_n\mathbb{E}\left\{Q_n(t)\mu_n(t)-V
\textrm{e}_{sw}^n (t)- V \sum_s \textrm{e}_s^n (t)
|\mathbf{Q}(t)\right\}}_{RHS^{ESS}}
 \label{eq:ess}
\end{align}
where $RHS^{ESS}$ is the right hand side of ~(\ref{eq:ess}). It is
now straightforward to see that ESS algorithm minimizes $RHS^{ESS}$
over all possible stationary algorithms.

We now present a fact given in~\cite{neely} that will be used to
prove that ESS is throughput optimal.
\begin{theorem} (Lyapunov Stability)
For scalar valued function $g(.)$ for which the minimum value is
equal to $g^*$, If there exists positive constants $V, \epsilon, B$
such that for all time slots $t$ and all backlog vectors
$\mathbf{Q}(t)$, the Lyapunov drift satisfies
\begin{align}
\Delta (t) +  V \mathbb{E}(g(t)|\mathbf{Q}(t)) \leq B-\epsilon
\sum_{i=1}^N Q_i(t) + V g^* , \label{eq:lemma2}
\end{align}
then the system is stable and the time average backlog satisfies:
\begin{align*}
 \limsup_{T\rightarrow \infty} \frac{1}{T} \sum_{t=0}^{T-1}
\sum_{i=1}^N \mathbb{E}\{Q_i(t)\} \leq \frac{B+V g^*}{\epsilon}
\end{align*}
\end{theorem}
\begin{theorem}
The proposed ESS algorithm is throughput optimal.
\end{theorem}
\begin{proof}
Clearly, the RHS of the randomized policy RND is given as follows
\begin{equation}
 RHS^{RND}= B - \sum_n\mathbb{E}\left\{Q_n(t)(\hat{\mu}_n(t)-R_n(t))|\mathbf{Q}(t)\right\}
\end{equation}
\noindent where $\hat{\mu}_n(t)$ is the transmission rate achieved by RND at time slot $t$.
By using~(\ref{eq:rnd_mu}) and adding the cost function,
we have
\begin{align}
B - \epsilon \sum_n Q_n(t) + &V \mathbb{E}
(\hat{H}_{\epsilon}(t)|\mathbf{Q}(t))\nonumber\\ &\leq B - \epsilon
\sum_n Q_n(t) + V H_{max} \label{eq:ra1}
\end{align}
where $\hat{H}_{\epsilon}(t)$ is the actual energy consumption of
the randomized policy during time slot $t$ and it depends on $\epsilon$.
By definition,
$RHS^{ESS} \leq RHS^{RND}$. Then, we obtain
\begin{align}
RHS^{ESS} \leq RHS^{RND}\leq B- \epsilon \sum_n Q_n(t) +  V
H_{max}\label{eq:ra2}
\end{align}
Thus
\begin{align}
\Delta(t) +  V \mathbb{E}(H(t)|\mathbf{Q}(t))\leq B- \epsilon \sum_n
Q_n(t) +  V H_{max}\label{eq:rason}
\end{align}
 \noindent It is straightforward to see that~(\ref{eq:rason})
has exactly same form as~(\ref{eq:lemma2}). Thus, this proves the
theorem.
\end{proof}
\subsection{Optimality of ESS} In this section, we show that ESS
algorithm yields asymptotically optimal solution to (6)-(7). In
other words, the average energy consumption proposed by the ESS
algorithm can be made arbitrarily close to the global minimum
solution of (6)-(7), by selecting a sufficiently large value of a
real-valued constant $V$. The following fact from~\cite{neely}
establishes a bound on any scalar-valued function $g(.)$.
\begin{theorem}(Lyapunov Optimization)
If~(\ref{eq:lemma2}) holds then the time average energy satisfies
the following,
\begin{align}
 \limsup_{T\rightarrow \infty} \frac{1}{T} \sum_{t=0}^{T-1}
 \mathbb{E}\{g(t)\} \leq g^* +  \frac{B}{V}
 \label{eq:lemma3}
\end{align}
\end{theorem}
\begin{theorem}
ESS algorithm is energy optimal and satisfies~(\ref{eq:lemma3}).
\end{theorem}
\begin{proof}
Rewriting~(\ref{eq:ra1}) and using~(\ref{eq:ra2}) we attain,
\begin{align}
\Delta (t) + V \mathbb{E}(H(t)|\mathbf{Q}(t)) \leq &B - \epsilon
\sum_n Q_n(t) \\&+ V \mathbb{E} (\hat{H}_{\epsilon}(t)|\mathbf{Q}(t))\nonumber
 \label{eq:ram}
\end{align}
We take the expectation of both sides of (29) and obtain
\begin{equation}
\Delta(t) + V \mathbb{E}(H(t)) \leq B - \epsilon \mathbb{E}
\left(\sum_n Q_n(t)\right) + V \mathbb{E} (\hat{H}_{\epsilon}(t))
 \label{eq:ra3}
\end{equation}
Note that inequality~(\ref{eq:ra3}) holds for any time slot $t$.
Hence, we sum both sides of~(\ref{eq:ra3}) from time slot $0$ to
$T-1$ and divide by $T$. As a result, we obtain
\begin{align}
&\frac{\mathbb{E} \left(\sum_n Q_n^2(T)\right)}{T}-\frac{\mathbb{E}
\left(\sum_n Q_n^2(0)\right)}{T} + \frac{V}{T} \sum_{t=0}^{T-1}
\mathbb{E}(H(t))\nonumber\\
 &\leq  B - \frac{\epsilon}{T}
\sum_{t=0}^{T-1} \sum_n \mathbb{E} (Q_n(t)) +  V \mathbb{E}
(\hat{H}_{\epsilon}(t))
 \label{eq:opt2}
\end{align}
Recall that in the previous section, we have shown that the second
term in the RHS of~(\ref{eq:opt2}) is bounded. By taking
$\limsup_{T\rightarrow \infty}$, we get
\begin{align}
\limsup_{T\rightarrow \infty} \frac{V}{T} \sum_{t=0}^{T-1} \sum_n
\mathbb{E} (H(t)) & \leq B +  V \mathbb{E} (\hat{H}_{\epsilon}(t))
 \label{eq:opt3}
\end{align}
Consequently, by dividing both sides of~(\ref{eq:opt3}) by $V$, we
have
\begin{align}
\limsup_{T\rightarrow \infty} \frac{1}{T} \sum_{t=0}^{T-1} \sum_n
\mathbb{E} (H(t)) & \leq \frac{B}{V} +  \mathbb{E}
(\hat{H}_{\epsilon}(t)) \label{eq:perf}
\end{align}
Note that $\mathbb{E} (\hat{H}_{\epsilon}(t)) \rightarrow h^*$ as
$\epsilon \rightarrow 0$ as we have shown in the previous section.
The bound in~(\ref{eq:perf}) is clearly minimized by taking a limit
as $\epsilon \rightarrow 0$, which yields~(\ref{eq:lemma3}).
\end{proof}
\subsection{Distributed Implementation} In this section, we present
an approximate solution to the scheduling problem (6)-(7) that does
not require a centralized mechanism and thus it is easy to implement
in a large-scale network. Recall that, the switching decision is
made by a centralized authority with ESS algorithm. However, in a
distributed way, each node can make a switching decision
individually according to the followings rules: 1-)  An active node
$n$ keeps staying in active mode as long as the following inequality
holds,
\begin{align}
Q_n(t) > & \frac{V(e_{1}^n(t)-e_0^n(t)-e_{10}^n(t))}{\mu_n(t)}, \quad \forall n \in Ac
\end{align}
Therefore, in distributed implementation there may be more than one
active nodes in the network at a given time slot. Note that, there
should be only one active node at each time slot with ESS algorithm,
and other nodes stay in sleep mode in order to reduce energy
consumption. Thus, we expect that ESS outperforms the distributed
algorithm in terms of lifetime. However, it is clear to see from
(10) that ESS algorithm requires a centralized authority to make a
switching decision since it requires the knowledge of the sleeping
node which maximizes the left hand side of (10).

2-) On the other hand, sleeping nodes remain in the sleep mode if
the following inequality is satisfied,
\begin{align}
Q_m(t) < & \frac{V(e_{1}^m(t)+e_{01}^m(t)-e_0^m(t))}{\mu_m(t)},
\quad \forall m \in Sl
\end{align}
\begin{algorithm}
 \algsetup{indent=1em}
\begin{algorithmic} [1]
 \STATE{\bf \footnotesize{Initialize}:} \footnotesize{Set all nodes to the sleep mode and the remaining energy $\scriptstyle E_n(0)=E,  n \in N, t=0$}
 \WHILE {$\scriptstyle E_i(t) >$ $0$}
 \STATE {\bf \footnotesize{generate}} \footnotesize{packet arrivals and transmission rate}
  \IF {$\textstyle i \in Sl $ and $  \textstyle Q_i(t) < \frac{V(e_{1}(t)+e_{01}(t)-e_0(t))}{\mu(t)} $}
   \STATE {\bf \footnotesize{compute}} $\scriptstyle E_i(t)=E_i(t)-e_0(t)$
   \STATE {\bf \footnotesize{compute}} $\scriptstyle weight(i)= - V e_0(t)$
 \ELSIF {$\textstyle i \in Sl $ and $  \textstyle Q_i(t)> \frac{V(e_{1}(t)+e_{01}(t)-e_0(t))}{\mu(t)} $}
 \STATE {\bf \footnotesize{add}} \footnotesize{node} $\scriptstyle i$ \footnotesize{to} $\scriptstyle Ac$ \COMMENT{\textit{\footnotesize{switch to active mode}}}
 \STATE {\bf \footnotesize{compute}} $\scriptstyle E_i(t)=E_i(t)-e_{01}(t)-c(t)$
 \STATE {\bf \footnotesize{compute}}  $\scriptstyle weight(i)=Q_i(t)\mu(t)- V (e_{01}(t)+e_1(t)+\alpha \mu(t))$
 \ENDIF
 \IF{$\textstyle i \in Ac $ and $  \textstyle Q_i(t < \frac{V(e_{1}(t)+\alpha \mu(t) - e_{10}(t)-e_0(t))}{\mu(t)} $}
 \STATE {\bf \footnotesize{add}} \footnotesize{node} $i$ \footnotesize{to} $\scriptstyle Sl$  \COMMENT{\textit{\footnotesize{switch to sleep mode}}}
 \STATE {\bf \footnotesize{compute}} $\scriptstyle E_i(t)=E_i(t)-e_{10}(t) - e_0(t)$
 \STATE {\bf \footnotesize{compute}} $\scriptstyle weight(i)= - V (e_{10}(t) + e_0(t))$
  \ELSIF {$\textstyle i \in Ac $ and $  \textstyle Q_i(t > \frac{V(e_{1}(t)+\alpha \mu(t) - e_{10}(t)-e_0(t))}{\mu(t)} $}
 \STATE {\bf \footnotesize{compute}} $\scriptstyle E_i(t)=E_i(t)-c(t)$
 \STATE {\bf \footnotesize{compute}} $\scriptstyle weight(i)=Q_i(t)\mu(t)- V (c(t)+\alpha \mu(t))$
 \ENDIF
 \IF {$\textstyle i \in Ac $}
 \STATE {\bf \footnotesize{compute}} $\scriptstyle E_i(t)=E_i(t)-e_b(t),
$ \COMMENT{\textit{\footnotesize{broadcast its weight}}}
 \IF {$ \textstyle i=\arg\max_j  weight(j) \quad j\in Ac$}
 \STATE $\scriptstyle E_{i}(t)=E_{i}(t)-\alpha(t) \mu_{i}(t)$
\ENDIF \ENDIF
  \STATE {\bf \footnotesize{update}} $\scriptstyle Q_i(t)$
  \STATE $\scriptstyle t=t+1$
 \ENDWHILE
\end{algorithmic}
\caption{\small{Distributed Algorithm}} \label{distributed}
\end{algorithm}
Otherwise, they switches to active mode. After determining the
energy modes, scheduling decision is given among the active nodes.
The knowledge of backlog levels in the neighboring nodes can be
maintained by infrequent broadcasting. The minimum energy
broadcasting method is proposed~\cite{Mario} and it can be applied
to our model in order to find the transmissions node. Afterwards,
the node which maximizes~(\ref{eq:sch}) is scheduled. A more formal
representation of distributed implementation is given below as
Algorithm 1.
\begin{equation}
 \max_{n\in Ac} \quad Q_n(t)\mu_n(t)-V \textrm{e}_{sw}^n (t)- V
e_{1}^n (t).
 \label{eq:sch}
 \end{equation}
\section{Numerical Results}
We consider a scenario where there are 5 sensor nodes and a base
station. The nodes observe their surrounding and generate random
packets and transmit them to the base station in a single-hop
fashion. We assume that each node has a battery with an initial
energy of 10 Joules. The arrival processes are assumed to be
Bernoulli processes with an average rate of 4 packets with size of
45 bytes per slot for all nodes and without loss of generality we
assume that a slot is 2 millisecond (ms) long. For a particular
node, there are three equally likely possible channel states,
i.e.,\emph{ Good, Medium, Bad} where the corresponding transmission
rate of a node is assumed to be 20, 12, and 5 packets respectively.
The parameters used in the simulation are given by~\cite{ruzzelli1}.
\begin{itemize}
  \item \textbf{\emph{sleep mode energy}}, $e_0$ : Nodes consume energy in sleep mode and it is equal to $0.015\times10^{-6} J $
  per millisecond.
  \item \textbf{\emph{active mode energy}}, $e_1$: The nodes in active mode consume
  energy due to being in active mode, $c(t)$ and also for the transmission
  of each packet, $\alpha(t)$. Thus, $e_1(t)=c(t) + \mu(t) \alpha(t)$. $c(t)$ is
  given as $36\times10^{-6} J$ per millisecond and $\alpha(t)$ is taken to be $30\times10^{-6}
  J$ per packet transmission.
  \item \textbf{\emph{energy due to switching from sleep to active}}, $e_{01}$ :
  Switching time from sleep to active is 0.7 ms and the switching energy from sleep to
  active is equal to $25.2 \times10^{-6} J$.
  \item \textbf{\emph{energy due to switching from active to sleep}}, $e_{10}$ : Switching time from active to sleep is 0.01 ms and the switching energy from active to
  sleep is equal to $2.85 \times10^{-6} J$.
  \item \textbf{\emph{Broadcast energy}}, $e_b(t)$: Broadcast energy is taken
  to be $8.33 \times10^{-8} J$ per bit.
\end{itemize}
If a sensor node is in sleep mode at the beginning of a time slot
and wants to switch to active mode, then the switching takes 0.7 ms.
Thus, the sensor can stay in active mode for a duration of 1.3 ms
since total slot duration is 2 ms second. Additionally, it consumes
a packet transmission energy if it is selected by ESS algorithm. As
a result, it consumes an amount of energy which is equal to $1.3 ms
\times c(t) + \alpha \mu(t)$. When the node keeps it energy mode
then, there will be no switching energy and the energy consumption
will be equal to only the sleep energy, $e_0(t)$. On the other hand,
if an active sensor is selected to be the active node the next slot,
then it can use entire slot for its own transmission since there
will be no switching time and the energy consumption will be $2 ms
\times c(t) + \alpha \mu(t)$. If it switches the sleep mode, the
energy spent by the node will be equal to $1.99 ms \times e_0(t) +
e_{10}$

First, we investigate the average energy consumption in sleep and
active modes for $V=500, 1000, 5000$, $10000,20000, 40000, 60000$
and $80000$. Figure~\ref{fig:eng} depicts the average energy spent
in active and sleep modes. As $V$ grows, the nodes run out their
batteries by staying in sleep mode most of the time. This prolongs
their lifetime since they consume the minimum energy in sleep mode
and this result agrees with Theorem 5. Due to the fact that the
nodes can not transmit in sleep mode, it yields an increase in queue
sizes of the nodes and consequently average delay in the network
increases.
 \begin{figure}
     \centering
     \scalebox{0.55}{\includegraphics{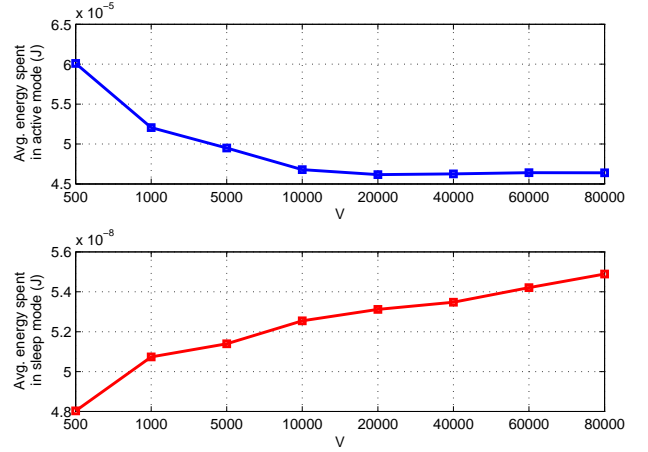}}
     \caption{\small{Avg. energy consumptions in active and sleep modes}}
     \label{fig:eng}
 \end{figure}
Next, we conduct a simulation experiment to observe the bursty
departure property of ESS algorithm explained in Lemma 1. As stated
in Lemma 1, the nodes transmit their packets consecutively in order
to reduce the switching energy. Figure~\ref{fig:bursty} depicts the
average number of bursty active staying which is the number of two
or more consecutive time slots where the node is selected to be the
active node and transmits. Higher values of $V$ push the nodes spend
the minimum energy and in order not to consume switching energy the
nodes are selected as the active node in a bursty fashion.
 \begin{figure}
     \centering
     \scalebox{0.45}{\includegraphics{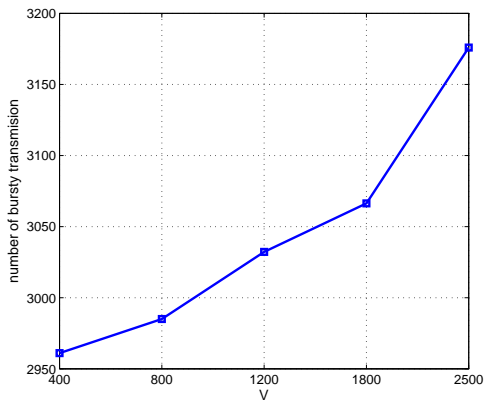}}
     \caption{\small{Bursty transmission}}
     \label{fig:bursty}
 \end{figure}
Furthermore, we compare the ESS algorithm with a ``benchmark
algorithm'' and the S-MAC type algorithm proposed in~\cite{smac}
,namely ``periodic algorithm'', in terms of the average queue
backlogs for different values of $V = 400, 800, 1200, 1800, 2500$.
Contrary to ESS, the benchmark algorithm ignores the switching
energy and makes the switching and scheduling decisions without
considering switching energy cost. Therefore, the nodes change their
modes more frequently. In S-MAC type \emph{periodic algorithm}, the
nodes stay in sleep mode during the fist 1 ms and during the last 1
ms, they stay in active mode. Furthermore, the node which maximizes
$Q_n(t) \mu_n(t) - V\alpha \mu_n(t)$ is selected as the transmitting
node during the active period. As it is seen in
Figure~\ref{fig:queue}, the benchmark algorithm outperforms ESS in
terms of the queue backlogs, since it stays in the active mode at
most of the time.
 \begin{figure}
     \centering
     \scalebox{0.50}{\includegraphics{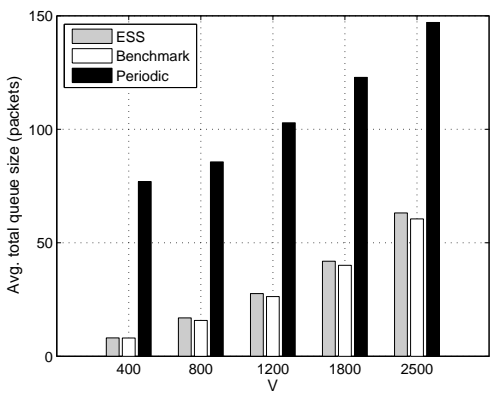}}
     \caption{\small{The average queue backlogs for different values of
     V}}
     \label{fig:queue}
 \end{figure}
On the other hand, since ESS promotes the nodes to stay in sleep
mode considering the switching cost, it prolongs the network
lifetime. The periodic algorithm has the worst performance since it
has fixed duty cycle property which can not ensure the efficient
time to reduce the queue backlogs as much as ESS algorithm.
Figure~\ref{fig:lifetime} depicts the comparison of the lifetime
performance induced by ESS, the benchmark, the distributed and the
periodic algorithms. The distributed algorithm allows more than one
nodes to be active at any time slot and the benchmark algorithm does
not care the switching energy cost. In addition, the periodic
algorithm forces the nodes to be active for 1 ms at every slot and
due tot he fixed duty cycle, $V$ does not effect the lifetime of
periodic algorithm. Thus, with these algorithms, the nodes stays in
active mode unnecessarily. As a result, batteries of the node run
out quickly. Therefore, EES shows better performance result in terms
of lifetime.
 \begin{figure}[]
     \centering
     \scalebox{0.50}{\includegraphics{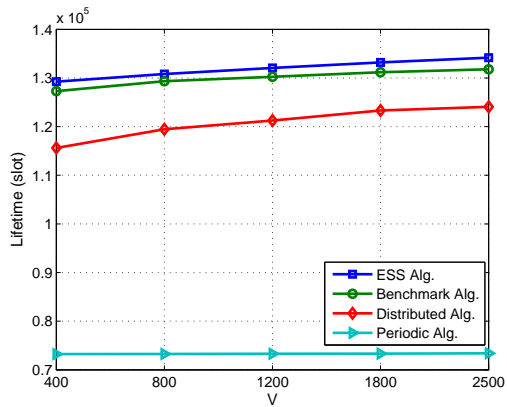}}
     \caption{\small{The network lifetime for different values of
     V}}
     \label{fig:lifetime}
 \end{figure}
Finally, Figure~\ref{fig:dutycycle} depicts the average duty cycle
of the network per node. As $V$ increases, the proportion of time
the nodes stay in active mode decreases to prevent the energy
consumption. Thus, the average duty cycle period decreases as well.
It is easy to see that the duty cycle per node can be approximately
at most \%13. As stated in~\cite{ruzzelli1}, the switching energy
needs to be computed if the node duty cycle is low. Thus, the effect
of switching energy becomes significant.
\begin{figure}[]
     \centering
     \scalebox{0.45}{\includegraphics{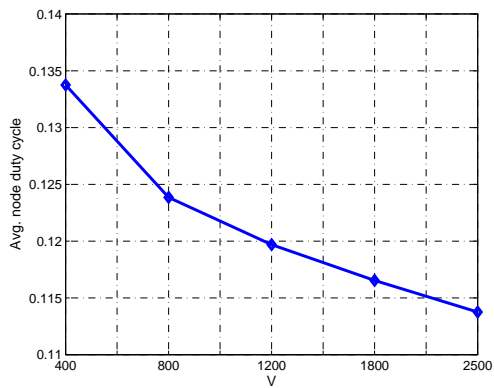}}
     \caption{\small{Duty cycle}}
     \label{fig:dutycycle}
 \end{figure}

\section{Conclusion}
In this paper, we investigate the energy efficiency and optimal
control problems in a low duty cycle wireless sensor network. Using
Lyapunov optimization technique, an energy-aware switching and
scheduling (ESS) algorithm is proposed. The switching energy is generally
not considered in energy management problems. However, analysis
indicate that ESS algorithm outperforms the benchmark algorithm
which ignores the switching energy cost in terms of the lifetime.
Furthermore, we show that we can stabilize the network by applying
ESS algorithm and push its performance to the optimal solution by
tuning $V$ parameter. In addition, we present a distributed algorithm that can be implemented easily.

In terms of future work, we investigate the power control problem.
Shortly, the power control problem is explained as follows.
If the sensor nodes have power control,
there is an inherent trade-off between the transmission power and
the time spent in active mode. Note that the transmission rate
$\mu_n(t)=\log\left(1+\frac{P_n(t)}{N}\right)$, where $P_n(t)$ and
$N$ are the transmission and ambient noise powers respectively, is
an increasing concave function. Therefore, as transmission power of
node $n$ increases, the number of bits sent increases and the
backlog of the active node drops below the aforementioned threshold
more quickly.

\end {document}